\documentclass[11pt]{amsart}
\usepackage{amscd}
\usepackage{verbatim}
\usepackage{amssymb}
\usepackage{amsmath}
\usepackage{amsthm}
\usepackage{amsfonts}
\usepackage{mathrsfs}
\usepackage{amsthm}
\usepackage{euscript}

\usepackage[margin=30 mm,heightrounded=true,centering]{geometry}

\newtheorem{theorem}{Theorem}[section]
\newtheorem{proposition}[theorem]{Proposition}
\newtheorem{lemma}[theorem]{Lemma}
\newtheorem{remark}[theorem]{Remark}
\newtheorem{corollary}[theorem]{Corollary}
\newtheorem{definition}[theorem]{Definition}

\begin{document}

\title[Bakry-\'Emery black holes]{Bakry-\'Emery black holes}

\author{M Rupert and E Woolgar}
\address{Department of Mathematical and Statistical Sciences,
University of Alberta, Edmonton, Alberta, Canada T6G 2G1}
\email{mrupert@ualberta.ca, ewoolgar(at)ualberta.ca}

\date{\today}

\begin{abstract}
\noindent Scalar-tensor gravitation theories,
such as the Brans-Dicke family of theories,
are commonly partly described by a modified Einstein equation in which the
Ricci tensor is replaced by the Bakry-\'Emery-Ricci tensor of a Lorentzian metric and scalar field. In physics this formulation is sometimes referred to as the ``Jordan frame''. Just as, in General Relativity, natural energy conditions
on the stress-energy tensor become conditions on the Ricci tensor,
in scalar-tensor theories expressed in the Jordan frame natural energy conditions become conditions on the Bakry-\'Emery-Ricci tensor. We show that,
if the Bakry-\'Emery tensor obeys the null energy condition with an
upper bound on the Bakry-\'Emery scalar function, there is a
modified notion of apparent horizon which obeys analogues of
familiar theorems from General Relativity. The Bakry-\'Emery
modified apparent horizon always lies behind an event horizon and the event horizon obeys a modified area theorem. Under more restrictive conditions, the modified apparent horizon obeys an analogue of the Hawking topology theorem in 4 spacetime dimensions. Since topological censorship is known to yield a horizon topology theorem independent of the Hawking theorem, in an appendix we obtain a Bakry-\'Emery version of the topological censorship theorem. We apply our results to the Brans-Dicke theory, and obtain an area theorem for horizons in that theory. Our theorems can be used to understand behaviour observed in numerical simulations by Scheel, Shapiro, and Teukolsky \cite{SST} of dust collapse in Brans-Dicke theory.
\end{abstract}

\maketitle

\section{Introduction}
\setcounter{equation}{0}

\noindent It is a fortunate circumstance that many of the tools of
Riemannian comparison geometry carry over to the Lorentzian case. In
particular, the basic scalar Riccati equation estimate for the mean
curvature of hypersurfaces becomes an estimate for the Raychaudhuri
equation, which leads to several important theorems in general
relativity, including singularity theorems and much of the general theory of
black holes and of cosmology as well.

In the past decade, a comparison theory for the Bakry-\'Emery-Ricci
tensor has emerged (\cite{WW}, see also \cite{Lott}). The
\emph{Bakry-\'Emery-Ricci tensor} (or simply \emph{Bakry-\'Emery
tensor}) arises in Riemannian manifolds that have a preferred scalar
function $f$ or, equivalently, a distinguished measure, and is defined by
\begin{equation}
\label{eq1.1} R^f_{ij}:=R_{ij}[g]+\nabla_i\nabla_j f\ ,
\end{equation}
where $R_{ij}[g]={\rm Ric}[g]$ is the Ricci tensor of the Riemannian
metric $g$ and $\nabla_i\nabla_j f ={\rm Hess} f$ is the Hessian
defined by the Levi-Civit\`a connection $\nabla$ of $g$.

Case \cite{Case} has shown that certain aspects of the Bakry-\'Emery
comparison theory carry over to the case of a Lorentzian metric,
allowing him to prove a Bakry-\'Emery version of the
Hawking-Penrose singularity theorem for general relativity. It was
observed in \cite{Woolgar} that this leads to a Hawking-Penrose
theorem for scalar-tensor gravitation theories in the so-called
Jordan frame formulation. It then seems reasonable to ask whether
more of the Bakry-\'Emery theory can be applied to the Lorentzian
case, and whether this can be used to bring the theory of black
holes in scalar-tensor theories to a state of development similar to
that which has been achieved in the case of pure general relativity.

Here we show that significant portions of the theory of black holes
can be adapted to a (Lorentzian) Bakry-\'Emery formulation. In
general relativity, energy conditions (positivity conditions on
components of the Ricci or Einstein tensor) lead to theorems that
govern apparent and event horizons. We show that when energy conditions
are applied instead to Bakry-\'Emery-Ricci tensor (or, as the case
may be, the Bakry-\'Emery version of the Einstein tensor), and when
$f$ obeys appropriate conditions as well, one can prove analogous
theorems governing the so-called $f$-modified apparent horizon.

We refer the reader to \cite{HE} and \cite{Wald} for the standard approach to the theory of black holes, including definitions and appropriate background needed below. One first needs the notion of an exterior region that is \emph{asymptotically flat} \cite[chapter 11]{Wald} and \emph{strongly asymptotically predictable} \cite[p 299]{Wald}. These exterior regions are called \emph{domains of outer communications} and are defined by
\begin{equation}
\label{eq1.2} {\mathcal D}:=I^+({\mathcal I}^-)\cap I^-({\mathcal I}^+)\ .
\end{equation}
Here ${\mathcal I}:={\mathcal I}^-\cup \{ i^0 \} \cup {\mathcal I}^+$
denotes a connected component of \emph{conformal infinity} (see
\cite[chapter 11]{Wald}), where ${\mathcal I}^+$ is future null infinity, ${\mathcal I}^-$ is past null infinity, and $i^0$ is spatial infinity. Also, $I^+(X)$ denotes the chronological future (past) of the set $X$; dually, $I^-(X)$ denotes the chronological past of $X$. Finally, we will augment asymptotic flatness by also requiring that the directional derivative $\nabla_l f$ of $f$ in any future-directed null direction $l$ must vanish on approach to ${\mathcal I}$.

Now consider a closed, spacelike, codimension 2 surface embedded in $(M,g,f)$. There are two linearly independent future-null vector fields orthogonal to this surface, denoted $l^{(i)}$ for $i\in \{1,2 \}$ (we sometimes write $l:=l^{(1)}$, $k:=l^{(2)}$ instead). The surface is called $f$-trapped if
\begin{equation}
\label{eq1.3} \theta^{(i)}_f:=\theta^{(i)}-\nabla_{l^{(i)}} f \le 0
\end{equation}
everywhere on the surface for each $i\in \{1,2 \}$, where
$\theta^{(i)}$ is the expansion scalar associated to $l^{(i)}$. We
will not always repeat the definitions of the terms \emph{trapped
surface}, \emph{marginally trapped surface}, \emph{outer trapped
surface}, \emph{apparent horizon}, etc, as these are well-explained
in standard texts, e.g., \cite[chapter 12]{Wald} and \cite[chapter
9]{HE}, but we will modify these terms with an $f$ (e.g.,
\emph{$f$-trapped surface}) to mean that the condition
$\theta^{(i)}=0$ or $\theta^{(i)}\le 0$, as appropriate, in the
conventional definition is replaced by $\theta^{(i)}_f=0$ or
$\theta^{(i)}_f\le 0$, respectively.

We obtain three general theorems related to well-known theorems in
general relativity. The spirit of these theorems is that what is
true for apparent horizons under standard energy conditions is true
for $f$-apparent horizons under Bakry-\'Emery energy conditions
together with additional conditions on $f$. The reasonableness of
these additional conditions is then the interesting question and one
for which considerations of physical applicability are important.
With this in mind, a fourth theorem casts our results in the realm
of scalar-tensor gravitation theory.

\begin{theorem}\label{theorem1.1}
Let ${\mathcal D}$ be an asymptotically flat and strongly
asymptotically predictable connected component of a spacetime
$(M,g,f)$. Assume that the $f$-modified null energy condition
\begin{equation}
\label{eq1.4} R^f_{ij}l^il^j\ge 0
\end{equation}
holds for all future-null vectors $l$ in $TM$ and that $\nabla_l
f\to 0$ on approach to ${\mathcal I}^+$. Assume also that there is a
$k\in {\mathbb R}$ such that $f\le k$. Then no closed $f$-trapped
surface and no $f$-apparent horizon intersects ${\mathcal D}$.
\end{theorem}

That is, $f$-apparent horizons must lie behind event horizons when
the conditions of the theorem hold.

In the standard theory, the Hawking area theorem arises from the
same basic analysis, essentially an estimate for the Raychaudhuri
equation, that underlies the proof that apparent horizons lie behind
event horizons. In the present case, we define the \emph{$f$-surface
area} or \emph{$f$-volume} of a closed spacelike submanifold $S$ as
\begin{equation}
\label{eq1.5} A_f[S]:=\int_S e^{-f}dS\ ,
\end{equation}
whenever the integral converges, where $dS$ is the volume element
induced on $S$ by the spacetime metric $g$. Then we have the
following result:

\begin{theorem}\label{theorem1.2}
In a strongly asymptotically predictable spacetime, assume that
$f\le k$ and that $R^f_{ij}l^il^j\ge 0$ for every null $l$. Let
${\mathcal H}$ be the black hole event horizon and let $\Sigma_1$
and $\Sigma_2$ be Cauchy surfaces intersecting ${\mathcal H}$ in
closed surfaces $S_1:=\Sigma_1\cap{\mathcal H}$ and
$S_2:=\Sigma_2\cap{\mathcal H}$, such that $S_2$ lies everywhere to
the future of $S_1$. Then $A_f[S_2]\ge A_f[S_1]$.
\end{theorem}

Apparent horizons in general relativity have a stability property
analogous to the stability of those minimal surfaces in Riemannian geometry which are genuine local minimizers of the area functional. In particular, there is a stability operator for apparent horizons, whose spectrum must be nonnegative.
In general relativity, this leads to the
Hawking horizon topology theorem \cite[theorem 9.3.2]{HE} (see \cite{GS} for a theorem in general dimension). Likewise in the present case there is a stability operator for $f$-apparent horizons and, in $4$ spacetime dimensions at least, an associated Hawking-type topology
theorem, when certain conditions are imposed on $f$. For example, we have

\begin{theorem}\label{theorem1.3}
Consider an $n=4$ dimensional spacetime. Assume that, for every
pair of future-timelike vectors $v$, $w$, the
$f$-modified Einstein tensor
\begin{equation}
\label{eq1.6} G^f_{\mu\nu}:=R^f_{\mu\nu}-\frac12 g_{\mu\nu}R^f\ ,
\end{equation}
(where $R^f:=g^{\mu\nu}R^f_{\mu\nu}$) obeys
\begin{equation}
\label{eq1.7} G^f_{\mu\nu}\big\vert_S v^{\mu}w^{\nu}\ge 0\ ,
\end{equation}
on an outer $f$-apparent horizon $S$, so
\begin{equation}
\label{eq1.8} \theta_f\big\vert_S:=\theta-\nabla_lf \big\vert_S=0\ ,
\end{equation}
where $l$ is the outbound null direction orthogonal to $S$. Assume
further that $f$ obeys
\begin{equation}
\label{eq1.9} \square f\big\vert_S \ge 0\ ,
\end{equation}
and
\begin{equation}
\label{eq1.10} \nabla_l f\big\vert_S=0 \ .
\end{equation}
Then every such outer $f$-apparent
horizon $S$ is either a 2-sphere or a torus with induced metric $e^{2f}\delta$, where $\delta$ is a flat metric.
\end{theorem}

Of course, although by (\ref{eq1.8}), any outer $f$-apparent horizon $S$ with
$\nabla_lf\big\vert_S=0$ has $\theta=0$, the original Hawking topology theorem cannot be directly applied. First, an outer $f$-apparent horizon need not be an outer apparent horizon but, moreover, the energy condition (\ref{eq1.7}) does not imply that $G_{\mu\nu}\big
\vert_S v^{\mu}w^{\nu}\ge 0$. As well, we draw attention to the borderline case of the toroidal horizon. In General Relativity, it was known that toroidal topology was only possible in the presence of the dominant energy condition if the induced metric on the horizon were flat. (Even this possibility was later ruled out entirely \cite{Galloway1} for outermost apparent horizons.) We see from Theorem \ref{theorem1.3} that in the present case the horizon metric is once again completely determined up to choice of flat torus metric $\delta$, but is now $e^{2f}\delta$ so in general it is not flat.

As the conditions on $f$ in Theorem \ref{theorem1.3} may appear not to be optimal, we should ask whether they are at least physical. To examine this question, we consider scalar-tensor gravitation theory. Scalar-tensor theories are in a sense almost ubiquitous in modern physics. They arise, for example, whenever a Kaluza-Klein type reduction from higher dimensions is employed, including in string theory. The prototypical scalar-tensor theory is Brans-Dicke gravitation theory formulated in $4$ spacetime dimensions, for which the above theorems lead to the following result:

\begin{theorem}\label{theorem1.4}
For some $\omega\in \left ( -\frac32,\infty \right )$,
assume that $(M^4,g,\varphi)$ is a solution of the Brans-Dicke
equations of gravitation in the Jordan frame formulation\footnote
{The Jordan frame formulation is that in which the field equations are as presented in section 5. There is also the Einstein frame formulation, in which the field equations resemble those of general relativity. These formulations are related by a conformal transformation. See \cite{Faraoni} for details.}
such that the matter stress-energy tensor $T_{\mu\nu}$
obeys $T_{\mu\nu}l^{\mu}l^{\nu}\ge 0$ for all future-null vectors
$l$. Define $f=-\log\varphi$ and assume that $\varphi\ge C >0$. Then
\begin{enumerate}
\item[(i)] any Jordan frame $f$-apparent horizon $S$ lies behind an event
    horizon,
\item[(ii)] any Einstein frame apparent horizon must lie behind an event horizon,
    and
\item[(iii)] for $S_1$, $S_2$ as in Theorem \ref{theorem1.2}, we have
\begin{equation}
\label{eq1.11} \int_{S_1}\varphi dS\le \int_{S_2}\varphi dS\ ,
\end{equation}
where $dS$ is the area element induced by $g$.\footnote
{That is, $dS$ is the \emph{Jordan frame} area element, so $\varphi dS$ is the \emph{Einstein frame} area element. Therefore, statement (iii) says that the Einstein frame horizon cross-section area increases.}
\end{enumerate}
Moreover, assume that $S$ is an outer $f$-apparent horizon such that
\begin{eqnarray}
\label{eq1.12} 0&\ge&T\big\vert_S:=g^{\mu\nu}T_{\mu\nu}\big\vert_S\ , \\
\label{eq1.13} 0&\le&\left [ T_{\mu\nu}-\frac12 g_{\mu\nu}
\frac{T}{(3+2\omega)}\right ]_S u^{\mu}v^{\nu}
\end{eqnarray}
for every choice of future-timelike vectors $u$, $v$ in $T_pM$,
$p\in S$, and assume that $\nabla_l f\big \vert_S \equiv
-\frac{1}{\varphi}\nabla_l \varphi \big \vert_S =0$, where $l$ is
the null vector in (\ref{eq1.8}). Then
\begin{enumerate}
\item[(iv)] $S$ is either a $2$-sphere or a $2$-torus, and if it is a $2$-torus then it has induced metric $\delta/\varphi^2$, where $\delta$ is a flat metric on the torus.
\end{enumerate}
\end{theorem}

Condition (\ref{eq1.12}) holds for perfect fluids of mass-energy
density $\rho$ and pressure $p$ if $\rho \ge 3p$, for free
scalar fields $\psi$ if $\left ( \frac{\partial \psi}{\partial
t}\right )^2 \le |\vec \nabla \psi |^2$, and for massless Maxwell
fields always. The condition that $\nabla_l\varphi\big\vert_S=0$ is
reasonable at least in the static case.

Parts (i--iii) of this theorem are in fact unsurprising. The conformal transformation from the Jordan to the Einstein frame formulation (see (\ref{eq2.10}), (\ref{eq2.12})) preserves the assumptions,
maps the Jordan frame horizon area element $dS$ to $\varphi dS$, and maps the Jordan frame $f$-modified horizon expansion scalar to the Einstein frame unmodified expansion scalar so as to preserve the sign (see Remark \ref{remark2.6}). In this light, parts (i--iii) are easily understood. They lend confidence to the underlying Bakry-\'Emery theorems on which they are based and which do not necessarily have a physical context. Furthermore,
the phenomenon in statement (ii) of Theorem \ref{theorem1.4} has been observed in an interesting numerical study of the collapse of collisionless dust in Brans-Dicke theory \cite{SST}. The stress-energy tensor of such a model obeys the null energy condition. However, the numerical evolution reveals that the null components of the Jordan frame Ricci tensor are sometimes negative \cite[figure 10]{SST}, and the apparent horizon with respect to the Jordan frame metric sometimes lies outside the event horizon \cite[figure 9]{SST}. Nonetheless, the Einstein frame apparent horizon always lies behind the event horizon \cite[figures 11 and 12]{SST}.

This paper is organized as follows. Section 2 contains the proof of Theorem
\ref{theorem1.1}, while Theorem \ref{theorem1.3} is proved in
section 3. Section 4 is devoted to Brans-Dicke theory and the proof
of Theorem \ref{theorem1.4}. An appendix is devoted to a simple
Bakry-\'Emery version of the topological censorship
theorem \cite{FSW}, based on an easy modification of an extant proof
\cite{GW}. This produces independent constraints on horizon topology as an immediate corollary (cf \cite{CW, GSWW}).

\medskip
\noindent\emph{Acknowledgements.} We thank GJ Galloway for
comments on a draft and for suggestions which improved the proof of
Proposition \ref{proposition2.4}. MR is grateful for an
Undergraduate Summer Research Award from the Natural Sciences and
Engineering Research Council (NSERC) of Canada. EW is grateful to the Park City Math Institute 2013 Research Program, during which some of this work was completed. This work was supported by an NSERC Discovery Grant to EW.

\section{$f$-trapped surfaces and $f$-apparent horizons}
\setcounter{equation}{0}

\noindent We begin with a
closed spacelike surface $S$ of codimension $2$ in a spacetime of
dimension $n$, thus $\dim S = n-2$. There are two congruences of
null geodesics issuing orthogonally from $S$, called the outgoing
and ingoing null congruences. As in the introductory section, we
denote the tangent fields to these congruences by $l^{(i)}$ and
their respective expansion scalars by $\theta^{(i)}$ with $i=1$ for
the outgoing congruence (we also use the notation
$\theta^{(1)}=:\theta$) and $i=2$ for the ingoing congruence (we
also write $\theta^{(2)}=\kappa$). Then $S$ is \emph{outer trapped}
if $\theta^{(i)}<0$ for $i=1$ and \emph{trapped} if
$\theta^{(i)}\vert_S<0$ for both $i=1$ and $i=2$. If
\begin{equation}
\label{eq2.1} \theta^{(i)}_f:=\theta^{(i)}-\nabla_{l^{(i)}} f<0
\end{equation}
then $S$ is \emph{outer $f$-trapped} if $i=1$ and \emph{$f$-trapped}
if $i=1,2$. We say $S$ is \emph{marginally $f$-trapped} if the
strict inequality in (\ref{eq2.1}) is replaced by the closed
inequality $\le 0$.

The key to the issue is the following lemma (see also \cite{Case}):

\begin{lemma}\label{lemma2.1}
Assume that $f$ is a smooth function obeying $f\le k$ for some
constant $k$ and say that
\begin{equation}
\label{eq2.2} {\rm Ric}^f(l,l)\ge 0
\end{equation}
along every null geodesic $\gamma$ with tangent field
$l=\frac{d}{ds}$ and belonging to a null geodesic congruence issuing
from $S$. If
$\theta_f<0$ on $S$, then there is a focal point to $S$ at some
$s>0$ along $\gamma$.
\end{lemma}
\begin{remark}
\label{remark2.2} The result holds for timelike $\gamma$ too,
provided (\ref{eq2.2}) holds for every \emph{timelike} geodesic
congruence issuing from $S$, with tangent field $l$.
\end{remark}
\begin{proof} The Riccati equation governing $\theta$ is known as the
Raychaudhuri equation. Dropping a term of definite sign, it yields
the inequality
\begin{equation}
\label{eq2.3} \frac{d\theta}{ds} \le
- {\rm Ric}(l,l) -\frac{\theta^2}{m}\ ,
\end{equation}
where $m=n-2$ for a null congruence ($m=n-1$ for a timelike
congruence). Adding $-\nabla_l\nabla_l f$ to each side and writing
${\rm Ric}^f:=\left ( {\rm Ric}+\nabla \nabla f\right )(l,l)= {\rm
Ric}(l,l)+ \nabla_l \nabla_l f$ (using that $\nabla_ll=0$), then
\begin{equation}
\begin{split}
\label{eq2.4} \frac{d\theta_f}{ds} \le&\,\,
- {\rm Ric}^f(l,l) -\frac{\theta^2}{m}\\
=&\,\, - {\rm Ric}^f(l,l)
-\frac{1}{m}\left [ \theta_f^2+2\theta_f \nabla_l f
+\left ( \nabla_l f\right )^2\right ]\\
\le&\,\, -\frac{\theta_f^2}{m}-\frac{2\theta_f \nabla_l f}{m}\ ,
\end{split}
\end{equation}
where on the left-hand side we used the definition of $\theta_f$ as
in (\ref{eq2.1}) and in the last line we used (\ref{eq2.2}). For
economy of notation, we suppress the geodesic $\gamma$ when it
occurs in the composition $f\circ \gamma$ and simply write $f$. The
inequality in passing from the middle line on the right-hand side to
the last line is an equality iff, along the null geodesic $\gamma$
under consideration, ${\rm Ric}^f(l,l)=0$ and $\nabla_l f=0$ (the
first line is an equality iff the shear of the congruence containing
$\gamma$ vanishes).

Then we can rewrite (\ref{eq2.4}) as
\begin{equation}
\label{eq2.5} \frac{d}{ds}
\left ( \frac{m}{e^{\frac{2f}{m}}\theta_f}\right )
\ge e^{-\frac{2f}{m}}\ .
\end{equation}
Now assume that $f$ is bounded above by $k$. Then, integrating
(\ref{eq2.5}) over $s\in [0,t]$, for $t$ small enough so that the left-hand side of (\ref{eq2.5}) remains well-defined, we obtain
\begin{equation}
\label{eq2.6}
\frac{m}{e^{\frac{2f(t)}{m}}\theta_f(t)} \ge \int\limits_0^t
e^{-\frac{2f(s)}{m}}ds+\frac{m}{e^{\frac{2f(0)}{m}}\theta_f(0)}
\ge te^{-\frac{2k}{m}}+\frac{m}{e^{\frac{2f(0)}{m}}\theta_f(0)}\ .
\end{equation}
This is an equality at $t=0$, and then there is an interval $t\in
[0,T)$ such that each side is negative since $\theta_f(0)<0$ for an
$f$-trapped surface. There will then be some $t=T_1>0$ such that the
right-hand side of (\ref{eq2.6}) approaches $0$ from below as
$t\nearrow T_1$. Thus the left-hand side must tend to $0$ from below
as $t\nearrow T_2$ for some $0<T_2\le T_1$. Since $f\le k$, then
necessarily $\theta_f(t)\to -\infty$ as $t\nearrow T_2$. But because
$\theta_f(t):=\theta-\nabla_l f $ and $f$ is smooth, then $\theta\to
-\infty$ as $t\nearrow T_2$. \end{proof}

\begin{corollary}\label{corollary2.3}
In a strongly asymptotically predictable and asymptotically flat
spacetime, no $f$-trapped surface meets ${\mathcal D}$.
\end{corollary}

The proof is standard:

\begin{proof}
Say $S$ is $f$-trapped. By way of contradiction, if $S\cap {\mathcal
D}$ were nonempty, there would exist causal curves from $S$ to
${\mathcal I}$. Since these causal curves cannot approach spatial
infinity $i^0$, the boundary $\partial I^+(S)$ therefore would meet
${\mathcal I}$, say at a point $q\in {\mathcal I}$. Since spacetime
is strongly asyptotically predictable, then $\partial I^+(S)$ would
contain a past-null geodesic generator beginning at $q\in {\mathcal
I}$ and which does not end until it reaches $S$. Reversing
direction, this gives a future-inextendible null geodesic
$\gamma:[0,\infty)\to M$ lying always on $\partial I^+(S)$. It
therefore cannot contain a focal point to $S$, contradicting Lemma
\ref{lemma2.1}.
\end{proof}

The case of a marginally trapped surface is somewhat
more subtle, and it is here that the asymptotic flatness of
$f$-strongly asymptotically predictable domains is needed.
\begin{proposition}\label{proposition2.4}
Let ${\mathcal D}$ be an asymptotically flat, strongly
asymptotically predictable domain of outer communications and assume
the conditions of Lemma \ref{lemma2.1} hold on ${\mathcal D}$.
Assume further that $\nabla_l f\to 0$ on approach to ${\mathcal
I}^+$ for any null $l$. Then no marginally $f$-trapped closed
spacelike surface $S$ intersects ${\mathcal D}$.
\end{proposition}
\begin{proof}[Sketch of proof.]
The standard proof is given, e.g., in \cite[Proposition 12.2.3, p
310]{Wald} (see also \cite[section 9.2, p 320]{HE}), but has a small
technical gap, so we follow instead the proof given in
\cite[Theorem 6.1]{CGS}. We restrict ourselves to outlining the
logic of that proof and discussing the modification required to
accommodate non-zero $f$.

All versions of the proof proceed by way of contradiction by
assuming that the future of $S$ meets ${\mathcal I}$. As this future
does not meet a neighbourhood of spatial infinity $i^0$, there is
then a point $q\in {\mathcal I}^+\cap \partial I^+(S)$. Using the
asymptotic predictability, $q$ can be reached from $S$ by a null
geodesic $\gamma$, which has no focal point to $S$ except possibly
at $q$ (this latter possibility is not accounted for older versions
of the proof).

In \cite{CGS}, a smooth spacelike surface $S^+\subset {\mathcal I}$
through $q$ is constructed. No point of this surface lies in the
chronological future of $S$ (at least, for suitably chosen $q$ and
$S^+$: in the presence of focal points on ${\mathcal I}^+$, the
choice of $q$ and $S^+$ in explained in \cite{CGS}). The boundary of
the past $\partial I^-(S^+)$ of this surface will contain $\gamma$.
The expansion scalar, call it $\theta^{(1)}$, of $\partial I^-(S^+)$
along $\gamma$ is computed using the unnumbered equation in the
proof of Theorem 6.1 in \cite{CGS}, where it is shown to be positive
and bounded away from $0$. That is, there is an $s_0\in {\mathbb R}$
such that we can write $\theta^{(1)}\big\vert_{\gamma(s)} \ge C>0$
for any $s\ge s_0$.

On the other hand, by equation (\ref{eq2.6}), the $f$-modified
expansion, say $\theta^{(2)}_f(t)$ of the outbound null congruence
from $S$ obeys
\begin{equation}
\label{eq2.7} \theta^{(2)}_f(t)
\le \frac{me^{-\frac{2f(t)}{m}}}{\frac{m}{\theta^{(2)}_f(0)}
e^{-\frac{2f(0)}{m}}+te^{-\frac{2k}{m}}}\le 0\ ,
\end{equation}
and so since $\nabla_l f\to 0$ near ${\mathcal I}$ (i.e., as
$t\to\infty$ along the congruence), then the ordinary expansion
scalar obeys $\theta^{(2)}(t)\le \epsilon$, where we can make
$\epsilon>0$ be arbitrarily small by choosing $t$ sufficiently
large. In particular, eventually
\begin{equation}
\label{eq2.8} \theta^{(2)}\le \epsilon<C\le\theta^{(1)}\ .
\end{equation}

The problem is that $I^+(S)$ cannot intersect $I^-(S^+)$ though the
boundaries of these sets share a common generating curve $\gamma$,
so the boundaries of $I^+(S)$ and $I^-(S^+)$ must ``bend away from
each other'' along $\gamma$. For this to happen, we must have
$\theta^{(2)}\ge \theta^{(1)}$, contradicting (\ref{eq2.8}). The
resolution of this contradiction is that $\partial I^+(S)$ cannot
meet ${\mathcal I}^+$, establishing the proposition.\end{proof}

In \cite{CGS}, the argument in the last paragraph above instead
proceeds by appeal to the maximum principle. The same coud have been
done here. In the \cite{CGS} (i.e., no $f$) case, one obtains
\ref{eq2.8} with $\epsilon=0$, implying both that $\theta^{(2)}\le
0<\theta^{(1)}$. The geometric maximum principle for smooth
hypersurfaces \cite[Theorem 2.1]{Galloway2}, then implies that
$\partial I^+(S)$ and $\partial I^-(S^+)$ would necessarily coincide
near $\gamma$ and would have expansion scalar
$\theta=\theta^{(1)}=\theta^{(2)}=0$, which contradicts another
implication of (\ref{eq2.8}) with $\epsilon=0$, which is that
$\theta^{(1)}\ge C>0$. In the present case (i.e., with $f$), because
$\epsilon>0$, we cannot appeal directly to \cite[Theorem
2.1]{Galloway2}. Nonetheless, from (\ref{eq2.8}) we have
$\theta^{(2)}\le \epsilon\le\theta^{(1)}$ and so, by appealing to
\cite[Theorem 2.2]{Galloway2}, we have
$\theta^{(1)}=\theta^{(2)}=\epsilon$ locally near $\gamma$. But
(\ref{eq2.8}) also implies that $\theta^{(1)}\ge C>\epsilon$ at
points along $\gamma$, so again we have a contradiction which
establishes the proposition.

Mimicking the standard analysis, we define a \emph{total $f$-trapped
region} to be the union of all $f$-trapped surfaces and define the
\emph{$f$-apparent horizon} to be its boundary. When this boundary
is smooth, it has outbound expansion $\theta_f^{(1)}=0$; the proof
is standard (e.g., \cite[Theorem 12.2.5, \emph{mutatis
mutandis}]{Wald}). Therefore the proof of Proposition
\ref{proposition2.4} applies in this case, so an $f$-apparent
horizon must lie entirely outside ${\mathcal D}$ (i.e., it coincides
with, or lies behind, a black hole event horizon). This completes
the proof of Theorem \ref{theorem1.1}.

We remark that, if the condition (\ref{eq2.2}) holds for all
timelike $l$ as well, if $f$ is bounded above, if an $f$-modified
version of the null generic condition holds, and if there are no
closed timelike curves, then Case's singularity theorem \cite{Case}
guarantees that the future of the trapped region is nonspacelike
geodesically incomplete.

We now briefly turn attention back to event horizons and the proof
of Theorem \ref{theorem1.2}. A well-known result in the standard
theory is that, when $R_{ij}l^il^j\ge 0$ for all null $l$, then the
null geodesic generators of an event horizon have $\theta\ge 0$
\cite[Lemma 9.2.2]{HE}. Here we have an analogous result:

\begin{lemma}\label{lemma2.5}
Assume that $f\le k$ and the $f$-null energy condition (\ref{eq2.2})
holds as in Lemma \ref{lemma2.1}. Then the $f$-modified expansion of
the null geodesic generators of the event horizon ${\mathcal H}$ of
a future asymptotically predictable spacetime obey $\theta_f\ge 0$.
\end{lemma}

\begin{proof}[Sketch of proof.]
The proof is the same as that of \cite[Lemma 9.2.2]{HE} or
\cite[proof of Theorem 12.2.6, first paragraph]{Wald}, but since we
have condition (\ref{eq2.2}) we must use $\theta_f$. If
$\theta_f(p)<0$ at some $p\in {\mathcal H}$, then one can deform
${\mathcal H}$ slighty outward so as to intersect ${\mathcal D}$,
keeping $\theta_f<0$ somewhere, say at $p'$. But then $\partial I^+
(p')$ will intersect ${\mathcal I}^+$, say at $q$, and, using
asymptotic predictability, one can trace back from $q$ to construct
a null geodesic generator of $\partial I^+ (p')$ which extends from
$p'$ to $q$. This generator cannot have a focal point, contradicting
Lemma \ref{lemma2.1}.
\end{proof}

Then we are in a position to prove Theorem \ref{theorem1.2}.

\begin{proof}[Proof of Theorem \ref{theorem1.2}]
Again, following the standard proof (\cite[Proposition 9.2.7]{HE},
\cite[proof of Theorem 12.2.6, second paragraph]{Wald}), we observe
that there is an injective (but possibly not surjective) map
$\psi:S_1\to S_2$ by transport along the null geodesic generators of
${\mathcal H}$. Lemma \ref{lemma2.5} we have that $\theta_f\ge 0$
along these generators. But, differentiating equation (\ref{eq1.5})
along any of these null geodesic generators, say with tangent vector
$l=\frac{d}{ds}$, we have
\begin{equation}
\label{eq2.9} \frac{dA_f}{ds}=\int_S \left ( \theta-\nabla_l f
\right ) e^{-f}dS\equiv \int_S \theta_f e^{-f}dS\ .
\end{equation}
This shows that $A_f[S_1]\le A_f[\psi(S_1)]\le A_f[S_2]$.
\end{proof}

Proposition \ref{proposition2.4} can also be obtained from standard results in \cite{CGS}, \cite{HE}, and \cite{Wald} which rely on the usual Raychaudhuri equation (\ref{eq2.3}), without invoking the $f$-modified equation (\ref{eq2.4}) and the associated estimate that appears in the proof of Lemma \ref{lemma2.1}. We no outline that argument. Let
\begin{equation}
\label{eq2.10} {\tilde g}_{ij}:=e^{-\frac{2f}{n-2}}g_{ij}\ .
\end{equation}
Then the expansion scalar of the $l$-congruence transforms as
\begin{equation}
\label{eq2.11} {\tilde \theta}=e^{\frac{2f}{(n-2)}}\left ( \theta - \nabla_l f \right ) =e^{\frac{2f}{(n-2)}}\theta_f\ .
\end{equation}

\begin{remark}\label{remark2.6}
$\theta_f$ and ${\tilde \theta}$ are either both positive, both negative, or both zero.
\end{remark}

The prefactor $e^{\frac{2f}{(n-2)}}$ ensures that ${\tilde \theta}$ is the expansion scalar of a congruence of null geodesics, not just pregeodesics, with respect to ${\tilde g}$. The conditions $f\le k$ and $\nabla_l f\to 0$ in the above theorems now ensure that the conformal rescaling (\ref{eq2.10}) preserves asymptotic flatness and future asymptotic predictably (e.g., unless $f\le k$, $g$-complete geodesics might not be ${\tilde g}$-complete).

Now ${\tilde \theta}$ is governed by the ordinary Raychaudhuri equation (without $f$-terms) for the rescaled metric ${\tilde g}$. A standard conformal transformation formula shows that the Ricci curvature ${\tilde R}_{ij}$ of ${\tilde g}$ is given in terms of the Ricci curvature $R_{ij}$ of $g$ by
\begin{equation}
\label{eq2.12}
\begin{split}
{\tilde R}_{ij}=&\,\, R_{ij}+\nabla_i\nabla_j f+\frac{1}{(n-2)}\left [ \nabla_i\varphi \nabla_j\varphi +g_{ij}\left (\Delta f -|\nabla f|^2 \right ) \right ]\\
=&\,\, R^f_{ij}+\frac{1}{(n-2)}\left [ \nabla_i\varphi \nabla_j\varphi +g_{ij}\left (\Delta f -|\nabla f|^2 \right ) \right ] \ .
\end{split}
\end{equation}
From this it's easy to see that, for $l$ a null vector, then $R^f_{ij}l^i l^j\ge 0$ implies ${\tilde R}_{ij}l^i l^j\ge 0$. Thus, our assumptions on $g$ imply that the null energy condition holds for ${\tilde R}_{ij}$ at every point, and then the usual analysis shows that apparent horizons in ${\tilde g}$ lie behind event horizons. But by (\ref{eq2.10}) $f$-trapped or marginally $f$-trapped regions with respect to $g$ are trapped or, respectively, marginally trapped with respect to ${\tilde g}$, and since event horizons are conformally invariant, it follows that $f$-apparent horizons with respect to $g$ lie behind event horizons.

\section{Stability and the Hawking topology theorem}
\setcounter{equation}{0}

\noindent Consider now an arbitrary variation $\Phi:S\times I \to M$
of a closed spacelike co-dimension 2 surface $S$ embedded in
spacetime $M$, where $I\subseteq {\mathbb R}$ is an open interval
containing $0$. We let $x_p(\sigma):=\Phi (p,\sigma)$ be the image
of $(p,\sigma)$, $p\in S$. Varying $\sigma$, then this yields a
curve such that $x(0)=p\in S$, and we define
$q:=\frac{\partial}{\partial \sigma}$ to be the tangent field to
this curve. We can specify the variation by specifying $q$. We write
\begin{equation}
\label{eq3.1}
\begin{split}
q:=&\,\, q^{\|}+q^{\perp}\ ,\\
q^{\perp}:=&\,\, bl-\frac12 uk\ ,
\end{split}
\end{equation}
where $q^{\|}$ is tangent to the leaves $S_{\sigma}$ of the
variation, $l$ and $k$ are linearly independent future null vectors
in the normal bundle $NS_{\sigma}$ to the leaves (defined first on
$S$, then parallel transported to a neighbourhood of $S$---see
\cite{AMS} for further details) and are normalized so that $k\cdot
l=-2$, and $b,u:S\to {\mathbb R}$ are arbitrary functions on $S$. We
define the \emph{vector second fundamental form}
$K:TS_{\sigma}\times TS_{\sigma}\to NS_{\sigma}$ by
\begin{equation}
\label{eq3.2} K(X,Y)=-\left ( \nabla_X Y\right )^{\perp}\ ,
\ X,Y\in TS_{\sigma}\ ,
\end{equation}
where the superscript $\perp$ denotes projection into $NS_{\sigma}$.
The \emph{mean curvature vector} of $S_{\sigma}$ is then $H={\rm
tr}_h K\equiv h^{AB}K_{AB}$, where $h_{\mu\nu}:=g_{\mu\nu}+\frac12
\left ( k_{\mu}l_{\nu}+l_{\mu}k_{\nu} \right )$ is the \emph{first
fundamental form} or \emph{induced metric} on $S$. The \emph{null
expansion scalars} are $\theta^{(1)}\equiv\theta:=H\cdot l$ and
$\theta^{(2)}\equiv \kappa:=H\cdot k$, and we have
$H= -\frac12 \left ( \theta k + \kappa l \right )$.

Then the variation of the expansion scalar $\theta$ along integral
curves of $q$ is given (in the notation of \cite{AMS}) by
\begin{equation}
\label{eq3.3}
\begin{split}
\delta_q \theta=&\,\, a\theta+q^{\|}(\theta)-b\left (
K_{AB}^{\mu}K^{\nu AB} +G^{\mu\nu} \right )l_{\mu}l_{\nu}\\
&\,\, -\left [ \Delta_{S_{\sigma}} -2s^AD_A -\left (
\frac12 R_{S_{\sigma}}-\frac12 H^2 -G_{\mu\nu} l^{\mu}w^{\nu}
+VG_{\mu\nu}l^{\mu}l^{\nu} -s^As_A+D_As^A\right ) \right ] u\ .
\end{split}
\end{equation}
Here $G_{\mu\nu}$ is the spacetime Einstein tensor, $R_{S_{\sigma}}$
is the scalar curvature of $S_{\sigma}$,
$\nabla_{S_{\sigma}}:=D^AD_A$ is the Laplacian on $S_{\sigma}$,
$D_A$ is the connection on $S_{\sigma}$ capital Latin indices run
over a basis for $TS_{\sigma}$, $\{ e_A \}$ is such a basis (say
orthonormal), $s_A:=-\frac12 k_{\mu} \nabla_{e_A} l^{\nu}$, and
$a:=-\frac12 k_{\mu} \frac{\partial}{\partial \sigma}
\big\vert_{\sigma=0} l^{\mu}$. As well, we have changed basis for
$NS_{\sigma}$ from $\{l,k\}$ to $\{l,w\}$, where
\begin{equation}
\label{eq3.4} w:=V l + \frac12 k \ .
\end{equation}
with $V\ge 0$ an arbitrary function on $S$, so that $w$ is future-causal. We will also make use of the combination
\begin{equation}
\label{eq3.5} v:=V l -\frac12 k\ .
\end{equation}
If $V\neq 0$ then $v$ is spacelike. Finally, since $H= -\frac12
\left ( \theta k + \kappa l \right )$, then $H^2 = -\theta\kappa$.

We can choose the variation such that $a=0$, $u=0$, $q^{\|}=0$, and
$b=1$. This yields the Raychaudhuri equation
\begin{equation}
\label{eq3.6}
\begin{split}
\nabla_l\theta =&\,\, - \left ( G_{\mu\nu}
+K_{\mu AB}K_{\nu}{}^{AB}\right )l^{\mu}l^{\nu}\\
=&\,\, - \left ( R_{\mu\nu} +K_{\mu AB}K_{\nu}{}^{AB}\right )
l^{\mu}l^{\nu}\\ \ge &\,\, - R_{\mu\nu}l^{\mu}l^{\nu}
- \frac{1}{m}\theta^2\ ,
\end{split}
\end{equation}
upon dropping a term of definite sign (cf equation (\ref{eq2.3})). A different choice is $a=0$, $q^{\|}=0$, $b=uV$, with $u:S\to {\mathbb R}$ to be chosen later, so that now $q=uv$. This yields
\begin{equation}
\label{eq3.7}
\begin{split}
\delta_{uv}\theta =&\,\, -\Delta_{S_{\sigma}}u +2s^AD_Au\\
&\,\, +\left ( \frac12 R_{S_{\sigma}} +\frac12 \theta\kappa
-G_{\mu\nu} l^{\mu}w^{\nu}
-VK_{\mu AB}K_{\nu}{}^{AB}l^{\mu}l^{\nu}-s^As_A+D_As^A\right ) u\ .
\end{split}
\end{equation}
This should be compared to equation (5) in \cite{AMS}, but note we
have been using $w$ in place of the vector $u$ in \cite{AMS} to
avoid confusion with the function $u$ Also, here we keep the
$-\frac12 H^2u\equiv \frac12\theta\kappa u$ term which is obviously
zero at a marginally trapped surface (so this term is dropped in
\cite{AMS}). We do this because, in the present case, we want
instead to consider marginally outer $f$-trapped surfaces, where
$\theta_f=0$. As with equation (5) in \cite{AMS}, from here onward
we will evaluate (\ref{eq3.7}) only on the surface $S$ (so
$\sigma=0$), which is an \emph{outer $f$-apparent horizon}, meaning
that it is the smooth boundary of a region whose points lie on outer
$f$-trapped surfaces, and thus $S$ is outer marginally $f$-trapped.
Then we have
\begin{equation}
\label{eq3.8}
\begin{split}
\delta_{uv}\big\vert_S \theta=&\,\,  -\Delta_S u +2s^AD_Au\\
&\,\, +\left ( \frac12 R_S +\frac12 \theta\kappa
-G_{\mu\nu} l^{\mu}w^{\nu} -VK_{\mu AB}K_{\nu}{}^{AB}l^{\mu}l^{\nu}
-s^As_A+D_As^A\right )\bigg \vert_S u\ .
\end{split}
\end{equation}

Now the idea is to replace various terms with $f$-modified terms.
Here and in what follows, we remind the reader that subscripts $l$
and $k$ refer to the null basis vectors and are not indices ranging
over a set of values. We begin with
\begin{equation}
\label{eq3.9}
\delta_{uv}\theta \equiv u\nabla_v\theta
= u\nabla_v \theta_f +u\nabla_v\nabla_l f\ .
\end{equation}

Moving to the right-hand side of (\ref{eq3.8}), we define
\begin{equation}
\label{eq3.10}
\begin{split}
G^f_{\mu\nu}:=&\,\, R^f_{\mu\nu}-\frac12 g_{\mu\nu}R^f
:= R_{\mu\nu}+\nabla_{\mu}\nabla_{\nu} f -\frac12 g_{\mu\nu}
\left ( R +\square f \right )\\
=&\,\, G_{\mu\nu}+\nabla_{\mu}\nabla_{\nu} f
-\frac12 g_{\mu\nu}\square f\ ,
\end{split}
\end{equation}
where we write $\square f$ for the d'Alembertian of $f$ (i.e., for
the scalar Laplacian $\square f=\nabla^{\mu}\nabla_{\mu}f$ in
Lorentzian signature).

Then
\begin{equation}
\label{eq3.11}
\begin{split}
G_{\mu\nu}l^{\mu}w^{\nu} =&\,\, G_{\mu\nu}^f l^{\mu} w^{\nu}
-\left ( \nabla_{\mu}\nabla_{\nu}f\right )l^{\mu}w^{\nu}
-\frac12 \square f\\
=&\,\, G_{\mu\nu}^f l^{\mu} w^{\nu}-\nabla_w \nabla_l f
+\nabla_w l \cdot \nabla f- \frac12 \square f\\
=&\,\, G_{\mu\nu}^f l^{\mu} w^{\nu}-\nabla_w \nabla_l f
+\frac12 \nabla_kl\cdot \nabla f -\frac12 \square f\ .
\end{split}
\end{equation}
Here we used that $\nabla_l l=0$.

If we insert (\ref{eq3.11}) and (\ref{eq3.9}) into (\ref{eq3.8}),
using (\ref{eq3.5}) and (\ref{eq3.4}) to combine the $\nabla_v
\nabla_l f$ term from (\ref{eq3.9}) with the $\nabla_w\nabla_l f$
term from (\ref{eq3.11}), we obtain
\begin{equation}
\label{eq3.12}
\begin{split}
u\nabla_v\big \vert_S \theta_f
=&\,\, -\Delta_S u +2s^AD_A u + \bigg [ \frac12 R_S
+\frac12 \theta\kappa -G^f_{\mu\nu}l^{\mu}w^{\nu}+ \frac12
\left ( k^{\mu}l^{\nu} +l^{\mu}k^{\nu} \right ) \nabla_{\mu} \nabla_{\nu}f\\
&\,\, +\frac12\nabla_kl\cdot\nabla f+\frac12\square f
-VK_{\mu AB}K_{\nu}{}^{AB}l^{\mu}l^{\nu} -s_As^A+D_A s^A \bigg ]_S u\\
=&\,\, -\Delta_S u +2s^AD_A u + \bigg [ \frac12 R_S
+\frac12 \theta\kappa -G^f_{\mu\nu}l^{\mu}w^{\nu}
+ h^{\mu\nu} \nabla_{\mu} \nabla_{\nu}f\\
&\,\, +\frac12\nabla_kl\cdot\nabla f-\frac12\square f
-VK_{\mu AB}K_{\nu}{}^{AB}l^{\mu}l^{\nu} -s_As^A+D_A s^A \bigg ]_S u\ ,
\end{split}
\end{equation}
using that $\frac12 \left ( k^{\mu}l^{\nu} +l^{\mu}k^{\nu} \right )
\nabla_{\mu} \nabla_{\nu}f=h^{\mu\nu}\nabla_{\mu}\nabla_{\nu}f
-g^{\mu\nu}\nabla_{\mu}\nabla_{\nu}f
=h^{\mu\nu}\nabla_{\mu}\nabla_{\nu}f-\square f$.

The $h^{\mu\nu}\nabla_{\mu}\nabla_{\nu} f$ term
can also be expanded. It yields
\begin{equation}
\label{eq3.13}
\begin{split}
h^{\mu\nu}\nabla_{\mu}\nabla_{\nu} f =&\,\, h^{\mu\nu} \nabla_{\mu}
\left [ g_{\nu\alpha}\nabla^{\alpha}f\right ]
=h^{\mu\nu} \nabla_{\mu}
\left [ \left ( h_{\nu\alpha}-\frac12 \left ( l_{\nu}k_{\alpha}
+k_{\nu}l_{\alpha} \right ) \right ) \nabla^{\alpha} f\right ]\\
=&\,\, \Delta_S f -\frac12 \left [ \left ( h^{\mu\nu} \nabla_{\mu}l_{\nu}
\right ) \nabla_k f +\left ( h^{\mu\nu} \nabla_{\mu}k_{\nu} \right )
\nabla_l f \right ]\\
=&\,\, \Delta_S f +\frac12 \theta \nabla_k f+\frac12 \kappa\nabla_l f\\
=&\,\, \Delta_S f +\nabla_k f \nabla_l f +\frac12 \theta_f \nabla_k f
+\frac12 \kappa_f\nabla_l f\ .
\end{split}
\end{equation}

To make further progress, we define
\begin{equation}
\label{eq3.14} K^f_{\mu AB}:= K_{\mu AB} - \frac{1}{m}h_{AB}
\nabla_{\mu} f\ ,
\end{equation}
where $m=n-2={\rm rank} (h_{AB})$. This will not be used to replace
the $K_{\mu AB}K_{\nu}{}^{AB}l^{\mu}l^{\nu}$ term, which is already
a square. Rather, we take the trace and define
\begin{eqnarray}
\label{eq3.15} H^f_{\mu}&:=&H_{\mu}-\nabla_{\mu} f\ , \\
\label{eq3.16}\theta_f&:=& H^f\cdot l\ ,\\
\label{eq3.17}\kappa_f&:=& H^f\cdot k \ ,
\end{eqnarray}
so $\theta_f$ and $\kappa_f$ are the $f$-modified expansions of the
$l$ and $k$ null congruences, respectively. Then observe that the
$\frac12\theta\kappa$ term  in (\ref{eq3.12}) becomes
\begin{equation}
\label{eq3.18} \frac12 \theta\kappa = \frac12 \left (
\theta_f+\nabla_l f \right )
\left ( \kappa_f+\nabla_k f \right ) \ .
\end{equation}

Inserting (\ref{eq3.13}) and (\ref{eq3.18}) into (\ref{eq3.12}),
we obtain
\begin{equation}
\label{eq3.19}
\begin{split}
u\nabla_v\big \vert_S \theta_f =&\,\, -\Delta_S u +2s^AD_A u
+ \bigg [ \frac12 R_S+\frac32 \nabla_l f \nabla_k f-G^f_{\mu\nu}l^{\mu}w^{\nu}
-\frac12 \square f\\
&\,\, -VK_{\mu AB}K_{\nu}{}^{AB}l^{\mu}l^{\nu} -s_As^A+D_A s^A +\Delta_S f\\
&\,\, +\frac12 \theta_f\kappa_f+\kappa_f\nabla_l f+\theta_f\nabla_k f
+\frac12 \nabla_k l \cdot \nabla f\bigg ]_S u\\
=&\,\, -\Delta_S u +2s^AD_A u + \bigg [ \frac12 R_S+\left (
\frac32\nabla_k f+\kappa_f\right ) \nabla_l f -G^f_{\mu\nu}l^{\mu}w^{\nu}
-\frac12 \square f\\
&\,\, -VK_{\mu AB}K_{\nu}{}^{AB}l^{\mu}l^{\nu} -s_As^A+D_A s^A +\Delta_S f
\bigg ]_S u \ .
\end{split}
\end{equation}
where in the last equality we used that $\nabla_k l \big\vert_S=0$
by construction (see \cite{AMS}) and $\theta_f\big\vert_S=0$. Thus
we obtain the \emph{$f$-stability operator}
\begin{eqnarray}
\label{eq3.20} {\tilde L}^f_w\psi&:=&\left [ L^f_w -\frac12 \square f
+\left ( \frac32\nabla_k f+\kappa_f\right ) \nabla_l f\right ]\psi\\
\label{eq3.21} L^f_w \psi&:=&-\Delta_S \psi +2s^AD_A \psi + \bigg [
\frac12 R_S -G^f_{\mu\nu}l^{\mu}w^{\nu}-VK_{\mu AB}K_{\nu}{}^{AB}l^{\mu}l^{\nu}\\
&& \nonumber -s_As^A+D_A s^A +\Delta_S f
\bigg ]_S \psi \ .
\end{eqnarray}
Here $L_w^f$ is the stability operator $L_w$ of \cite[equation
(5)]{AMS} with $G_{\mu\nu}$ replaced by $G^f_{\mu\nu}$ and with the
divergence term $D_as^A$ modified to become $D_A s^A+\Delta_S
f=D_A\left (s^A+D^Af\right )$ (as well, \cite{AMS} writes $L_v$, not $L_w$). It is noted in \cite{AMS} that
operators of this form, though not self-adjoint, have a real
\emph{principal eigenvalue} $\lambda_0\le {\rm Re}(\lambda)$ where
$\lambda$ is any other eigenvalue, and the eigenfunction
corresponding to $\lambda_0$ is positive. The operator ${\tilde
L}_v^f$ is further modified as in (\ref{eq3.20}) to produce the
$f$-stability operator.

\begin{proof}[Proof of Theorem \ref{theorem1.3}]
Under assumptions (\ref{eq1.9}, \ref{eq1.10}), we have
\begin{eqnarray}
\label{eq3.22} {\tilde L}^f_w\psi &\le& L^f_v\psi
\le-\Delta_S\psi +2s^AD_A\psi  +
\left [ Q^f -s_As^A+D_A s^A \right ] \psi\\
\label{eq3.23}Q^f&:=&Q+\Delta_S f:=
\frac12 R_S -G^f_{\mu\nu}l^{\mu}w^{\nu}-VK_{\mu AB}K_{\nu}{}^{AB}l^{\mu}l^{\nu}
+\Delta_S f\ ,
\end{eqnarray}
when $\psi$ is a positive function. Following the reasoning in
\cite{GS}, we note that the surface $\Sigma$ cannot be $f$-trapped
unless the principal eigenvalue, say $\lambda_1$, of the operator on the right-hand
side of (\ref{eq3.22}) is positive. If $\phi_1>0$ belongs to the eigenspace (we can arrange that $\phi_1$ is positive), then using $\psi=\phi_1$ in (\ref{eq3.22}) and writing $u=\log\phi_1$, we get
\begin{equation}
\label{eq3.24} -\Delta_S u +Q+D_As^A-\vert s-D u \vert^2\ge 0\ ,
\end{equation}
where we have completed the square on the terms $2s^AD_A\psi-s^As_A$
in (\ref{eq3.22}).

If we fix the spacetime dimension be $n=4$ and integrate (\ref{eq3.24})
over $S$, which is now a closed $2$-surface $S$, and use of the
divergence theorem, then we obtain $\int_S R_S dS\ge 0$. Thus the
Euler characteristic of $S$ is nonnegative and is positive unless
the conditions $G_{\mu\nu}^fl^{\mu}w^{\nu}$, $s=Du$, and $K_{\mu AB}l^{\mu}=0$
all hold pointwise on $S$ and for all null $l$ and timelike $w$. Hence, $S$ is either a $2$-sphere or a $2$-torus.

It remains to prove that in the latter case, the induced metric is $e^{2f}\delta$. To do so, we follow the argument \cite{GS}, which is given in arbitrary dimension. We maintain arbitrary dimension for as long as possible, to gain inisight into that case.

Multiplying (\ref{eq3.24}) by $\psi^2$ for $\psi\in C^{\infty}(S)$ and using some simple identities, we have
\begin{equation}
\label{eq3.25}
\begin{split}Q^f\psi^2-\vert s-Du\vert^2 \ge& -\psi^2 D\cdot (s-Du)\\
=&-D\cdot\left ( \psi^2 (s-Du)\right )+2\psi (D\psi)\cdot (s-Du)\\
\ge & -D\cdot\left ( \psi^2 (s-Du)\right )-2|\psi|\vert D\psi\vert \vert s-Du \vert\\
\ge & -D\cdot\left ( \psi^2 (s-Du)\right )-\vert D\psi \vert^2 - \vert s-Du \vert^2 \psi^2\ .
\end{split}
\end{equation}
Re-arranging terms and integrating, we obtain
\begin{equation}
\label{eq3.26}\int_S \left ( Q^f \psi^2 +\vert D\psi \vert^2 \right ) \ge 0
\end{equation}
for any $\psi\in C^{\infty}(S)$. Now by the Rayleigh formula, the lowest eigenvalue ${\hat \lambda}_1$ of the self-adjoint operator ${\hat L}= -\Delta_S+Q^f$ is given by
\begin{equation}
\label{eq3.27}{\hat \lambda}_1=\inf_{\psi}\frac{\int_S \left ( Q^f \psi^2 +\vert D\psi \vert^2 \right )}{\int_S \psi^2}\ ,
\end{equation}
where the infimum is over all $\psi\in C^{\infty}(S)\backslash \{ 0 \}$ (i.e., excluding the zero function). Thus, by (\ref{eq3.26}), we see that ${\hat \lambda}_1 \ge 0$. Let ${\hat \phi}_1$ denote a corresponding eigenfunction, chosen so that ${\hat \phi}_1>0$.

As above, let $h_{AB}$ be the metric induced by $g_{\mu\nu}$ on $S$.
\begin{equation}
\label{eq3.28}{\hat h}_{AB}:=\varphi^{2/(n-3)}h_{AB}
:=\left ( e^{-f}{\hat \phi}_1 \right )^{2/(n-3)} h_{AB}\ .
\end{equation}
The scalar curvature ${\hat R}_S$ of ${\hat h}$ is given in terms of the scalar curvature $R_S$ of ${\hat h}$ by a standard formula:
\begin{equation}
\label{eq3.29}
\begin{split}
{\hat R}_S=&\,\, \varphi^{-\frac{2}{(n-3)}} \left \{ R_S-\frac{2}{\varphi}
\Delta_S\varphi +\left ( \frac{n-2}{n-3}\right ) \frac{|D\varphi |^2}{\varphi^2}
\right \}\\
=&\,\, \left ( e^{-f}{\hat \phi}_1 \right )^{-\frac{2}{(n-3)}} \left [ R_S
+2\Delta_S f -2\frac{\Delta_S {\hat \phi}_1}{{\hat \phi}_1}
+2 \frac{ \vert D{\hat \phi}_1 \vert^2}{{\hat \phi}_1^2} -\left (
\frac{n-4}{n-3} \right ) \left \vert \frac{D{\hat \phi}_1}{{\hat \phi}_1}
-Df\right \vert^2 \right ]\\
=&\,\, \left ( e^{-f}{\hat \phi}_1 \right )^{-\frac{2}{(n-3)}} \left [ 2{\hat \lambda}_1+2G^f(l,w)+2V \vert K(l)\vert^2
+2 \frac{ \vert D{\hat \phi}_1 \vert^2}{{\hat \phi}_1^2} -\left (
\frac{n-4}{n-3} \right ) \left \vert \frac{D{\hat \phi}_1}{{\hat \phi}_1}
-Df\right \vert^2 \right ] \ ,
\end{split}
\end{equation}
where in the middle equality we used that $\varphi=e^{-f}{\hat \phi}_1$  and in the final equality we used that $-\Delta {\hat \phi}_1+Q^f{\hat \phi}_1= {\hat \lambda}_1{\hat \phi}_1$. As well, we used (\ref{eq3.23}) and defined $\vert K(l)\vert^2 :=K_{\mu AB} K_{\nu}{}^{AB}l^{\mu}l^{\nu}$. Now ${\hat \lambda}_1\ge 0$. Thus, when $n=4$ and when the energy condition (\ref{eq1.7}) holds, we see that ${\hat R}_S\ge 0$ pointwise. Then $\chi(S)=0\Rightarrow {\hat R}_S\ge 0$,which in turn implies that ${\hat \lambda}_1=0$, $G^f\big\vert_S(l,w)=0$ for every future-null $l$ and future-timelike $w$, $K_{\mu AB}\big\vert_S l^{\mu}=0$
for every future-null $l$ outbound from $S$, and $D{\hat \phi}_1=0$. But then the middle line in (\ref{eq3.29}) collapses to
\begin{equation}
\label{eq3.30} 0=R_S+2\Delta_S f\ .
\end{equation}
But in $2$-dimensions, $e^{2f}\left ( R_S+\Delta_Sf \right )$ is the scalar curvature of the metric $e^{-2f}g$, and so $e^{-2f}g$ is a flat metric.
\end{proof}

\section{Scalar-tensor theory}
\setcounter{equation}{0}

\noindent The Brans-Dicke theory \cite{BD, Faraoni} is actually a
family of theories in $4$ spacetime dimensions parametrized by a real number $\omega\in \left (-\frac32,\infty\right )$, and containing a metric $g_{ij}$
and scalar field $\varphi>0$ as ``gravitational variables''. For a
given non-gravitational stress-energy tensor $T_{ij}$ whose trace is
$T:=g^{ij}T_{ij}$ and for a fixed value of $\omega$ these variables are solutions of the system
\begin{eqnarray}
\label{eq4.1} G_{ij}&=&\frac{1}{\varphi}\left ( \nabla_i\nabla_j\varphi
-g_{ij}\square \varphi \right )+ \frac{8\pi}{\varphi}T_{ij}
+\frac{\omega}{\varphi^2} \left ( \nabla_i \varphi \nabla_j \varphi
-\frac12 g_{ij} \vert \nabla \varphi \vert ^2\right )\ ,\\
\label{eq4.2} \square \varphi &=&\frac{8\pi T}{3+2\omega}\ .
\end{eqnarray}
The form of these equations is of course not invariant under
conformal transformations. This particular form is called the
\emph{Jordan frame} form of the theory in the physics literature.
Using
\begin{equation}
\label{eq4.3} f:=-\log \varphi
\end{equation}
and carrying out some straightforward manipulations including a
substitution using (\ref{eq4.2}), equation (\ref{eq4.1}) can be
brought to the form
\begin{eqnarray}
\label{eq4.4} G^f_{ij} &=& 8\pi e^f \left [ T_{ij}-\frac12 g_{ij}
\frac{T}{\left ( 3+2\omega \right )} \right ] +8\pi (1+\omega) S_{ij}\ ,\\
S_{ij} &=&\frac{1}{8\pi} \left ( \nabla_i f \nabla_j f -\frac12 g_{ij}
\vert \nabla f \vert^2 \right ) \ ,
\end{eqnarray}
where $G^f_{ij}$ is the Bakry-\'Emery Einstein tensor defined in
(\ref{eq3.10}) and $S_{ij}$ is the stress-energy tensor of a free
scalar field. We caution that $\vert \nabla f \vert^2
:=g^{-1}(df,df)$ can be negative since $g$ is Lorentzian.

\begin{proof}[Proof of Theorem \ref{theorem1.4}]
Let $v$ be a future-timelike vector and define $\tau_i:=S_{ij}v^j$.
Then it is easy to verify that $\tau_i\tau^i\le 0$ and $\tau_i
v^i\ge 0$, and as this is the case for any future-timelike $v$, then
$\tau$ is past-causal. This in turn implies that $S_{ij}u^iv^j\ge 0$
for any future-timelike (or, by continuity, future-null) $u$ and
$v$. Thus $S_{ij}$ obeys the \emph{dominant energy condition}. If as
well $T_{\mu\nu}$ obeys condition (\ref{eq1.13}) for all
future-timelike vectors $u$ and $v$, we then have
$G^f_{\mu\nu}u^{\mu}v^{\nu}\ge 0$.

\medskip

\noindent\emph{Proof of (i):} From the last paragraph, it follows
that $S_{ij}$ obeys the null energy condition $S_{ij}l^il^j\ge 0$
for all null vectors $l$. By assumption, $T_{ij}$ also obeys the
null energy condition. Then from (\ref{eq4.4}) and (\ref{eq1.6}), we
see that $R^f_{ij}l^il^j \ge 0$ for all null vectors $l$.
Furthermore, $\varphi\ge C>0\Rightarrow f\le \log (1/C)$ so $f$ is
bounded above. Then, invoking Theorem \ref{theorem1.1}, we see that
any outer $f$-apparent horizon must lie behind an event horizon,
proving part (i) of Theorem \ref{theorem1.4}.

\medskip

\noindent\emph{Proof of (ii):} We see from (\ref{eq2.12}) that the conformal transformation (\ref{eq2.10}) eliminates the Hessian term from the Ricci tensor and thus from the field equation (\ref{eq4.1}). Thus, the Einstein frame metric is the metric ${\tilde g}$ obtained from $g$ using the $n=4$ case of (\ref{eq2.10}). By Remark \ref{remark2.6}, the Jordan frame $f$-trapped and Einstein frame trapped regions correspond. Thus, the boundaries of these regions correspond. Then the result follows from part (i).

\medskip

\noindent\emph{Proof of (iii):} Invoke Theorem \ref{theorem1.2} and
observe that (\ref{eq4.3}) implies that $e^{-f}=\varphi$ so
$A_f[S]=\int_S \varphi dS$.

\medskip

\noindent\emph{Proof of (iv):} Now note that
\begin{equation}
\label{eq4.5} \square f = -\square \log \varphi = -\frac{1}{\varphi}
\square \varphi +\frac{\vert d\varphi \vert^2}{\varphi^2}
= -\frac{8\pi e^{-f} T}{3+2\omega}+\vert df\vert^2 \ge \vert df\vert^2
\end{equation}
when $T\big\vert_S\le 0$ and $\omega\ge -3/2$. Because $|\cdot|$ is a Lorentzian norm, $\vert df\vert^2$ could be negative.
However, it isn't, since $\nabla_l\varphi \big \vert_S=0\Leftrightarrow \nabla_l f \big\vert_S=0 \Rightarrow \vert df\vert^2 \equiv h^{AB}D_AfD_Bf\big
\vert_S -\nabla_l f \nabla_k f \big \vert_S = h^{AB}D_AfD_Bf\big
\vert_S \ge 0$. Thus $\square f\big\vert_S
\ge 0$, and we can invoke Theorem \ref{theorem1.3}.
\end{proof}

\appendix

\section{Bakry-\'Emery topological censorship}
\setcounter{equation}{0}

\noindent We recall that the principle of topological censorship
governs the topology of black hole horizons:

\begin{proposition}\label{propositionA.1}
Consider a $4$-dimensional, asymptotically flat spacetime with a
globally hyperbolic connected component ${\mathcal D}$ of the domain
of outer communications and non-empty event horizon. If the
principle of topological censorship holds connected component of the
event horizon is a $2$-sphere.
\end{proposition}

The \emph{principle of topological censorship}, also called
\emph{active topological censorship}, states that any causal curve
in an asymptotically flat domain of outer communications ${\mathcal
D}$ and beginning and ending on ${\mathcal I}$ is fixed-endpoint
homotopic to a causal curve on ${\mathcal I}$. A theorem \cite{FSW}
states that the principle holds whenever a globally hyperbolic
${\mathcal D}$ obeys the \emph{null energy condition} or the weaker
\emph{averaged null energy condition} (ANEC).\footnote
{For other formulations, including modified formulations for
asymptotically anti-de Sitter spacetimes, see \cite{GSWW}.}

\begin{proof}[Proof of Proposition \ref{propositionA.1}]
See \cite{CW, GSWW}.\end{proof}

The ANEC
may be stated as the requirement that
\begin{equation}
\label{eqA.1} \int\limits_{-\infty}^{\infty}
{\rm Ric}(\eta',\eta')d\lambda\ge 0
\end{equation}
along every affinely parametrized complete null geodesic
$\eta:(-\infty,\infty)\to M$, where the integral in (\ref{eqA.1}) is
taken over an affine parameter. In (\ref{eqA.1}), by ``$\ge 0$'', we
include the case that the integral diverges, provided in the latter
case the improper integral in (\ref{eqA.1}), when replaced by an
integral over $[a,b]$, is nonnegative for all choices of $a$ and $b$
with $a\le -c$, $b \ge c$, for some $c$.

If we replace the Ricci tensor in (\ref{eqA.1}) by the Bakry-\'Emery
tensor and use that $\eta$ is geodesic so that
$\nabla_{\eta'}\eta'=0$, then the left-hand side is replaced by
\begin{equation}
\label{eqA.2} \int\limits_{-\infty}^{\infty}{\rm Ric}^f(\eta',\eta')d\lambda
= \int\limits_{-\infty}^{\infty}{\rm Ric}(\eta',\eta')d\lambda+\nabla_{\eta'}
f \big\vert^{\infty}_{-\infty}\ .
\end{equation}
Indeed, more generally, given a vector field $w$ we can use instead
the \emph{harmonic Ricci tensor}
\begin{equation}
\label{eqA.3}{\rm Ric}_w[g]:={\rm Ric}[g]+\frac12 \pounds_w g\ ,
\end{equation}
which reduces to the Bakry-\'Emery tensor when $w=\nabla f$, and
then
\begin{equation}
\label{eqA.4} \int\limits_{-\infty}^{\infty}{\rm Ric}_w(\eta',\eta')d\lambda
= \int\limits_{-\infty}^{\infty}{\rm Ric}(\eta',\eta')d\lambda+\eta'\cdot w
\big\vert^{\infty}_{-\infty}\ .
\end{equation}
This suggests the following definition, which we formulate so as to
include both infinite and semi-infinite geodesics:
\begin{definition}\label{definitionA.2}
Given a future-directed causal geodesic $\eta:I\to M$ with
$I=(-\infty,\infty)$, we say a vector field $w$ is \emph{net
decreasing with respect to $\eta$} if $\lim_{\lambda\to +\infty}
\eta'(\lambda)\cdot (w\circ \eta)(\lambda)\le \lim_{\lambda\to
-\infty} \eta'(\lambda)\cdot (w\circ \eta)(\lambda)$. If instead
$I=[0,\infty)$, we say a vector field $w$ is \emph{net decreasing
with respect to $\eta$} if $\lim_{\lambda\to +\infty}
\eta'(\lambda)\cdot (w\circ\eta)(\lambda)\le \eta'(0)\cdot (w\circ
\eta)(0)$.
\end{definition}
%

%
%\begin{remark}\label{remark2.2}
Net decreasing vector fields include various special cases of
interest. For example, if $w$ vanishes on a connected component
${\mathcal I}$ of conformal infinity then it is net decreasing along
any causal geodesic $\eta:(-\infty,\infty)\to M$ beginning and
ending on ${\mathcal I}$. In an asymptotically flat spacetime, $w$
is also net decreasing along any $\eta:(-\infty,\infty)\to M$ if it
is future-causal at ${\mathcal I}^+$ and past-causal at ${\mathcal
I}^-$. In cases such as those just discussed, we say that \emph{$w$
is net decreasing on the domain of outer communications ${\mathcal
D}$}. If ${\mathcal D}$ has a Cauchy surface $\Sigma$ such that $w$
is past-causal (including possibly vanishing) along it and
future-causal (again, possibly vanishing) at ${\mathcal I}^+$, then
$w$ is net decreasing along any null geodesic $\eta:[0,\infty)\to M$
from $\eta(0)\in \Sigma$ to ${\mathcal I}^+$, and then we say that
\emph{$w$ is net decreasing in the future development of $\Sigma$}.
%\end{remark}
%

Therefore, along any geodesic $\eta:I\to M$ along which $w$ is net
decreasing, we have
\begin{equation}
\label{eqA.5} \int_I{\rm Ric}(\eta',\eta')d\lambda \ge
\int_I{\rm Ric}_w(\eta',\eta')d\lambda\ .
\end{equation}
We are thus led to a second definition:

\begin{definition} \label{definitionA.3}
We say that the \emph{$w$-averaged null energy condition} (or
\emph{$w$-ANEC}) is obeyed along an infinite or semi-infinite null
geodesic $\eta:I\to M$, $I=[0,\infty)$ or $I=(-\infty,\infty)$, if,
for a vector field $w$ defined along $\eta$, we have
\begin{equation}
\label{eqA.6} \int_I{\rm Ric}_w(\eta',\eta')d\lambda\ge 0\ .
\end{equation}
We say that a spacetime $(M,g)$ obeys the \emph{$w$-averaged null
energy condition} if the $w$-ANEC is obeyed along every null
geodesic $\eta:(-\infty,\infty)\to M$. We say that a spacetime
$(M,g)$ obeys the \emph{$w$-averaged null energy condition to the
future of a Cauchy surface $\Sigma$} if
\begin{equation} \label{eqA.7} \int\limits_{0}^{\infty}{\rm
Ric}_w(\eta',\eta')d\lambda\ge 0
\end{equation}
along all complete future null geodesics $\eta:[0,\infty)\to M$ such
that $\eta(0)\in \Sigma$. In the special case that $w$ is a gradient
vector field $w=\nabla f$, we also say that $(M,g)$ obeys the
\emph{$f$-averaged null energy condition} (or \emph{$f$-ANEC}).
\end{definition}

Note that these conditions reduce to the familiar ANEC or averaged
null energy conditions when $w\equiv 0$. Also note that (\ref{eq2.2}) implies (\ref{eqA.6}) or (\ref{eqA.7}) (keeping in mind the remarks after (\ref{eqA.1})).
Combining these definitions, we immediately have the following
lemma:
\begin{proposition}\label{propositionA.4}
Let $w$ be a net decreasing vector field on a domain of outer
communications ${\mathcal D}$ of an asymptotically flat spacetime
$(M,g)$ which is causally continuous at spatial infinity $i^0$ and
such that every null geodesic visible from ${\mathcal I}^+$ is
future-complete in $(M,g)$. If the generic curvature condition (as
defined in \cite{GW} or \cite{HE}) and $w$-weighted ANEC hold along
every null geodesic $\eta:(-\infty,\infty)\to {\mathcal D}$, then
the topological censorship theorem, in the form of \cite[Theorem
2]{GW}, holds on ${\mathcal D}$.
\end{proposition}
\begin{proof}
By the assumptions and by (\ref{eqA.5}), we have that ANEC holds in
${\mathcal D}$. By \cite[Remark 2.2]{GW}, then every complete null
geodesic in ${\mathcal D}$ will have a pair of conjugate points.
Then the conditions of \cite[Theorem 2]{GW} hold. \end{proof}

This version of topological censorship is sometimes viewed as being
less than fully satisfactory because of its reliance on the generic
curvature condition. We can remove this assumption, but then we must
instead assume that $w$ is net decreasing from a Cauchy surface and
that the $w$-ANEC condition holds to the future of that Cauchy
surface. This then yields the version of topological censorship
found in \cite{FSW}, see also \cite{GSWW}, with the ANEC assumption
replaced by the $w$-ANEC condition and the net decreasing condition
for $w$.


\begin{thebibliography}{99}
\bibitem{AMS} L Andersson, M Mars, and W Simon, \emph{Stability
    of marginally outer trapped surfaces and existence of marginally
    outer trapped tubes}, Adv Theor Math Phys 12 (2008) 853--888.
\bibitem{BD} C Brans and RH Dicke, {\it Mach's principle and a
    relativistic theory of gravitation}, Phys Rev 124 (1961)
    925--935.
\bibitem{Case} JS Case, \emph{Singularity theorems and the
    Lorentzian splitting theorem for the Bakry-Emery-Ricci tensor}, J
    Geom Phys 60 (2010) 477--490.
\bibitem{CGS} PT Chru\'sciel, GJ Galloway, and D Solis, \emph{Topological censorship for Kaluza-Klein spacetimes}, Annales Henri Poincare 10 (2009) 893--912.
\bibitem{CW} PT Chru\'sciel and RM Wald, \emph{On the topology of
    stationary black holes}, Class Quantum Gravit 11 (1994)
    L147--L152.
\bibitem{Faraoni} V Faraoni, \emph{Cosmology in scalar-tensor gravity} (Kluwer, Dordrecht, 2004).
\bibitem{FSW} JL Friedman, K Schleich, and DM Witt, \emph{Topological censorship}, Phys Rev Lett 71 (1993) 1486--1489; Erratum-ibid 75 (1995) 1872.
\bibitem{Galloway1} GJ Galloway, \emph{Rigidity of marginally trapped surfaces and the topology of black holes}, Commun Anal Geom 16 (2008) 217--229.
\bibitem{Galloway2} GJ Galloway, \emph{Maximum principles for null
    hypersurfaces and null splitting theorems}, Ann Henri Poincar\'e
    1 (2000) 543--567.
\bibitem{GSWW} GJ Galloway, K Schleich, DM Witt, and E
    Woolgar, \emph{Topological Censorship and Higher Genus Black
    Holes}, Phys Rev D60 (1999) 104039.
\bibitem{GS} GJ Galloway and R Schoen, \emph{A generalization of
    Hawking's topology theorem to higher dimensions}, Commun Math
    Phys 266 (2006) 571--576.
\bibitem{GW} GJ Galloway and E Woolgar, \emph{The cosmic censor
    forbids naked topology}, Class Quantum Gravit 14 (1997) L1--L7.
\bibitem{HE} SW Hawking and GFR Ellis, \emph{The large scale
    structure of space-time} (Cambridge University Press, Cambridge, 1973).
\bibitem{Lott} J Lott, {\it Some geometric properties of the
    Bakry-\'Emery-Ricci tensor}, Comment Math Helv 78 (2003)
    865-–883.
\bibitem{SST} MA Scheel, SL Shapiro, and SA Teukolsky, \emph{Collapse to black holes in Brans-Dicke theory. II. Comparison with general relativity}, Phys Rev D51 (1995) 4236--4249.
\bibitem{Wald} RM Wald, {\it General relativity}, (University of
    Chicago Press, Chicago, 1984).
\bibitem{WW} G Wei and W Wylie, \emph{Comparison theory for the
    Bakry-\'Emery Ricci tensor}, J Diff Geom 83 (2009) 337--405.
\bibitem{Woolgar} E Woolgar, \emph{Scalar-tensor gravitation and the
    Bakry-Emery-Ricci tensor}, Class Quantum Gravit 30
    (2013) 085007.
\end{thebibliography}
\end{document}